\documentclass[11pt]{article}
\usepackage[letterpaper, margin=0.75in]{geometry}

\usepackage{graphicx,color}
\usepackage[cp850]{inputenc}
\usepackage[T1]{fontenc}
\usepackage{rotating}
\usepackage[f]{esvect}
\usepackage{amsmath, amsthm, amssymb}
\usepackage{rotating}
\usepackage{lscape}
\usepackage{mathrsfs}
\usepackage{subcaption}
\usepackage{mathtools}
\usepackage[toc,page]{appendix}
\usepackage{hyphenat}
\usepackage{mathptmx}
\usepackage[scaled=.65]{helvet}
\usepackage{courier}
\usepackage[toc,page]{appendix}
\usepackage{authblk}
\usepackage[right, mathlines]{lineno}
\usepackage{enumitem}

\usepackage[normalem]{ulem}

\def\tablenotes{\bgroup\parfillskip=0pt plus 1fil
\leftskip=0pt\relax \rightskip=0pt
\vskip2pt\footnotesize}
\def\endtablenotes{\vskip1pt\egroup}

\newtheorem{theorem}{Theorem}[section]
\newtheorem{proposition}[theorem]{Proposition}
\newtheorem{corollary}[theorem]{Corollary}

\newtheorem{lemma}[theorem]{Lemma}
\newtheorem{remark}[theorem]{Remark}

\renewcommand{\epsilon}{\varepsilon}
\renewcommand{\leq}{\leqslant}
\renewcommand{\geq}{\geqslant}
\renewcommand{\d}{\mathrm{d}}
 
\newcommand{\phat}{\hat{p}}

\renewcommand{\epsilon}{\varepsilon}

\newcommand{\TimeDeriv}{\frac{\textrm{d}}{\textrm{dt}}}

\allowdisplaybreaks

\usepackage[square,numbers,sort]{natbib}

\usepackage{graphicx}
\usepackage{overpic}
\usepackage{tikz}
 \graphicspath{ {./Figures/} }

\title{Reducing phenotype-structured partial differential equations models of cancer evolution to systems of ordinary differential equations: a generalised moment dynamics approach}

	\author[1,*]{Chiara Villa}
	\author[2]{Philip K Maini}
    \author[2]{Alexander P Browning}
    \author[3]{Adrianne L Jenner}
    \author[4,5]{Sara Hamis}
    \author[6,*]{Tyler Cassidy}

  \affil[1]{Sorbonne Universit\'e, CNRS, Universit\'e de Paris, Inria, Laboratoire Jacques-Louis Lions UMR 7598, 75005 Paris, France}
    \affil[2]{{Universit\'e Paris-Saclay, Inria, Centre Inria de Saclay, 91120, Palaiseau, France}}
    \affil[3]{Mathematical Institute, University of Oxford, Oxford, United Kingdom}
    \affil[4]{School of Mathematical Sciences, Queensland University of Technology, Brisbane, Australia}
    \affil[5]{Department of Information Technology, Uppsala University, Uppsala, Sweden}
	\affil[6]{Faculty of Medicine and Health Technology, Tampere University, Tampere, Finland}
    \affil[7]{School of Mathematics, University of Leeds, Leeds, United Kingdom}
    \affil[*]{Correspondence to: t.cassidy1@leeds.ac.uk (TC), chiara.villa@inria.fr(CV)}
\date{\today}

\begin{document}

\maketitle
 
 \noindent{\bf Author ORCiDs}: 
\smallskip
{\bf CV:} 0000-0003-3127-0532; {\bf PM:} 0000-0002-0146-9164;  {\bf APB:} 0000-0002-8753-1538; {\bf ALJ:} 0000-0001-9103-7092; \textbf{SH:} 0000-0002-1105-8078;   {\bf TC:} 0000-0003-0757-0017.

\section*{Abstract}
Intratumour phenotypic heterogeneity is understood to play a critical role in disease progression and treatment failure. Accordingly, there has been increasing interest in the development of mathematical models capable of capturing its role in cancer cell adaptation. 
This can be systematically achieved by means of models comprising phenotype-structured nonlocal partial differential equations, tracking the evolution of the phenotypic density distribution of the cell population, which may be compared to gene and protein expression distributions obtained experimentally. Nevertheless, given the high analytical and computational cost of solving these models, much is to be gained from reducing them to systems of ordinary differential equations for the moments of the distribution. We propose a generalised method of model-reduction, relying on the use of a moment generating function, Taylor series expansion and truncation closure, to reduce a nonlocal reaction-advection-diffusion equation, with general phenotypic drift and proliferation rate functions, to a system of moment equations up to arbitrary order. Our method extends previous results in the literature, which we address via two examples, by removing any \textit{a priori} assumption on the shape of the distribution, and provides a flexible framework for mathematical modellers to account for the role of phenotypic heterogeneity in cancer adaptive dynamics, in a simpler mathematical framework.

\clearpage

\section{Introduction}

Intratumour heterogeneity is increasingly understood as a primary determinant of disease progression and therapeutic response in solid cancers \citep{McGranahan2017,Burkhardt2022,Marine2020}. 
 While this heterogeneity has long-been viewed through the lens of clonal differences, recent experimental and clinical studies have implicated non-genetic heterogeneity as a driver of drug resistance and treatment failure \citep{Bell2020,Hanahan2022,Labrie2022}. The Epithelial-Mesenchymal Transition (EMT) is a well-studied example of non-genetic resistance \citep{Shi2023,Hanahan2022}, although a multitude of other examples exist, including adaptive rewiring of the mitogen activated protein kinase pathway \citep{Labrie2022}, and drug tolerant persisters in non-small cell lung cancer~\citep{Sharma2010,Goldman2015}. 
Alongside its role in the development of drug-resistance, non-genetic plasticity is at the core of metabolic and morphological changes in cancer cells that facilite their  metastatic spread and survival in harsh environments~\citep{mosier2021cancer,shen2020cell,tasdogan2020metabolic}.

Recent studies indicate that epigenetic regulation of genetically identical cancer cells induces a reversible drug-tolerant phenotype that expands during anticancer therapy \citep{Kavran2022,Shaffer2017}. In particular, transcriptomic data have identified reversible phenotypic changes that drive the development of resistance to targeted anti-cancer therapies and are mediated by a number of complex physiological factors \citep{Kavran2022,Shaffer2017}. Indeed, recent advances in multi-omics techniques have illustrated the complex dynamics of gene and protein expression that drive phenotypic plasticity \citep{Chen2023,Tirosh2016}. 
This ability to characterise population-level phenotypic plasticity permits a deeper understanding of evolution of non-genetic intratumour heterogeneity and the population distribution in phenotype space. Consequently, these experimental advances facilitate the development of mathematical models designed to capture both the shape and evolution of the phenotypic distribution of cancer cells.

Accordingly, there has been increased interest in the development of mathematical models to characterise the role of phenotypic heterogeneity in drug resistance and tumour progression~\citep{clairambault2019survey,marusyk2020intratumor}. A variety of deterministic and stochastic modelling frameworks have been proposed to study the evolutionary dynamics of phenotype-structured populations ~\citep{Cassidy2021,Gunnarsson2020,Anderson2006,stace2020discrete}.
Many existing models of phenotypic plasticity have focused on characterising the dynamics of a fixed and finite number of phenotypes, with transitions between these discrete states corresponding to an evolutionary game~\citep{Kareva2022,Craig2019,Kaznatcheev2019,West2018}. 
This discrete-phenotype framework is particularly common in the study of cancer treatment and relies on the assumption of the existence of drug-sensitive and drug-resistant subpopulations~\citep{Smalley2019,Craig2019,Cassidy2019}. 
However, as the role of continuously increasing levels of drug resistance has become apparent in driving treatment response, there has been increased interest in understanding the \textit{adaptive dynamics} that drive short-term phenotypic adaptation in response to, for example, the application of chemotherapeutic drugs. 
Moreover, the relevance of capturing phenotypic variants on a continuum extends beyond the study of the development of drug resistance, as phenotypic changes in cells are mediated by variations in the level of expression of relevant genes and proteins, which are indeed measured on a continuum. 

Dieckmann and Law~\citep{Dieckmann1996} proposed an adaptive dynamics framework to explicitly capture the dynamics of the continuous adaptation of the \textit{mean} phenotypic state, with population level heterogeneity captured by means of stochastic fluctuations in phenotype space, specifically focusing on `mutation-selection' dynamics which are easily transferable to the study of cancer~\citep{Aguade-Gorgorio2018,Martinez2021}. 
Moreover, their derivation of an ODE system for the dynamics of the mean evolutionary path inspired many deterministic studies, generally more amenable to analytical investigations~\citep{Gunnarsson2020,Altrock2015}, of adaptive dynamics relying on the simplifying assumption of a monomorphic population~\citep{Coggan2022,Martinez2021,vincent2005evolutionary}. 
These models often comprise a system of ordinary differential equations (ODEs) that model the evolution of both the population size and the mean trait $m_1(t)$. 
Nonetheless, as increasing importance is being attributed to population-level heterogeneity, we focus on deterministic frameworks providing a mean-field macroscopic description of stochastic, individual-cell dynamics, that do not rely on the limiting assumption of a monomorphic population. 

The resulting modelling approach typically describes the time-evolution of the entire phenotypic density function, denoted $p(t,x)$. The dynamics of the population, which is continuously structured by the variable $x$ in phenotype space, are modelled by a nonlocal partial differential equation (PDE) which typically takes the form of a reaction-advection-diffusion equation. 
Equations of this type may be studied with the theory of semigroups of operators, fixed point theorems in Banach spaces, and numerical methods, with semi-classical asymptotic methods having become increasingly popular to study certain limit cases ~\citep{chisholm2016effects,Dieckmann1996,dull2021spaces,perthame2006transport,perthame2008dirac}. 

This framework explicitly captures the continuous phenotypic adaptation of the population while preserving information on population-level heterogeneity, and is becoming increasingly common as experimental advances allow for direct characterisation of cancer cell phenotypes. For example, \citet{almeida2024evolutionary} and \citet{celora2022dna} leveraged time-resolved flow cytometry experiments to inform a structured PDE model of adaptation to nutrient or oxygen deprivation. 
Notably, in these works the phenotypic state $x$ is interpreted as the level of expression of a certain gene or protein, although this need not be the case \citep{Chisholm2015,Cho2018}. 

Although phenotype-structured PDEs carry increased biological relevance in the context of heterogeneous tumours, they pose a set of challenges that do not apply to standard ODE modelling frameworks.
For example, it is possible to establish the existence of equilibrium solutions of these PDE models via fixed point approaches in an appropriate Banach space. However, this fixed-point approach may not be constructive and is more mathematically involved than calculating equilibrium solutions of ODE models, which typically only involves solving a possibly non-linear set of equations. In addition, numerical methods for PDEs are typically implemented on a case-by-case basis, while highly efficient and accurate solvers for ODEs are found in most software packages. As parameter estimation typically involves many model simulations, the increased numerical efficiency of solving ODE is magnified when calibrating these models against experimental data. 
Multi-omics approaches that characterise the phenotypic distribution of tumours are becoming increasingly common. 
Nevertheless, experimental machines providing a full characterisation of gene and protein expression distributions (e.g. via flow cytometry or mass spectronomy) are yet to be widely available in experimental facilities due to their elevated cost, and more accessible techniques (e.g. Western blotting or RNA-seq) may only describe lower order moments -- such as the mean and variance -- of the phenotype distribution. 
Here, we
develop a technique to reduce phenotype-structured PDEs  to a system of ODEs for the moments characterising the phenotypic density function $p(t,x)$.
 This reduction will allow modellers to use existing technical tools for ODE models, such as those for identifiability analysis, sensitivity analysis, and model parameterisation, while maintaining the biological relevance and interpretability of the phenotype-structured PDE. A similar approach has been applied in mathematical oncology~\citep{almeida2024evolutionary,ardavseva2020evolutionary,lorenzi2015dissecting,villa2021evolutionary}, generally building on the model reduction procedure developed by \citet{almeida2019evolution} and \citet{chisholm2016evolutionary}. There, the authors showed that if the initial phenotypic density distribution, with $x\in\mathbb{R}$, is normally distributed with mean $m_1(0)$ and variance $\sigma^2(0)$, then under further restrictions on the population net-proliferation rate and phenotypic drift rates, it is possible to obtain explicit ODEs for $m_1(t)$ and $\sigma^2(t)$. However, while protein expression distributions may be approximately normal in some cases~\citep{almeida2024evolutionary}, the assumption that the phenotypic trait $x$ is unbounded and possibly negative is not, in general, compatible with biological data. 
Moreover, the additional restrictions on the functional forms of terms relating to proliferation rate and phenotypic drift in these studies reduce the model applicability to a limited range of biological scenarios. 

Here, we propose a generalised method to reduce phenotype-structured PDEs modelling cell adaptive dynamics to a system of ODEs for the moments characterising the phenotypic density function $p(t,x)$. 
Our method allows us to extend the analysis presented in~\citep{almeida2019evolution,chisholm2016evolutionary,lorenzi2015dissecting} by:
\begin{itemize}[noitemsep,topsep=0pt]
    \item[(i)] relaxing the assumption of an unbounded phenotype space, thus working in a more biologically relevant phenotypic domain;
    \item[(ii)] removing all \textit{a priori} assumptions on the shape of the distribution; 
    \item[(iii)] removing the additional restrictions on the phenotypic drift and net proliferation rate terms.
\end{itemize}
The model reduction procedure relies on the use of the moment generating function of the phenotypic distribution and techniques such as Taylor series expansion and moment closure that have previously been used in the stochastic modelling literature~\cite{Engblom:2006,Kuehn:2016uf,Fan:2016,Schnoerr.2017wbb,wagner2022quasi}. We thus obtain a system of ODEs for the moments characterising the phenotypic density function up to an arbitrary order.  The remainder of the manuscript is structured as follows. After introducing the general phenotype-structured reaction-advection-diffusion equation in Section~\ref{Sec:GeneralModel}, we demonstrate the model reduction procedure in Section~\ref{Sec:Analysis}, compare results with several examples from the extant literature in Section~\ref{Sec:Examples}, before concluding with a discussion in Section~\ref{Sec:Discussion}.

\section{A general phenotype-structured PDE model of cell adaptive dynamics} \label{Sec:GeneralModel}

Let $p(t,x)$ denote the phenotypic density function of the population at time $t$, i.e. the density of cells in the phenotypic state $x\in\Omega\subset\mathbb{R}$ at time $t\in\mathbb{R}_{\geq0}$, with $\Omega:=[{l},L]$ ($0<{l}<L$) a compact and connected set. The population size at time $t$, $P(t)$, is obtained by integrating over all possible phenotypes and is given by
\begin{equation}\label{Eq:PopulationSize}
P(t) = \int_\Omega p(t,x)\d x .
\end{equation}
We assume that $p(t,x)$ satisfies the following PDE
\begin{equation} \label{Eq:StructuredPDE}
\partial_t p(t,x) - \beta\partial^2_{xx} p(t,x) + \partial_x \big[\,V(t,x)p(t,x)\,\big] = \left( f(t,x)- \frac{P(t)}{\kappa} \right) p(t,x),
\end{equation}
for $t>0$ and $x\in\Omega$. Eq.~\eqref{Eq:StructuredPDE} is complemented with no flux boundary conditions 
\begin{equation}\label{Eq:BCs}
\beta \partial_x p + V(t,x)p = 0 \quad \textrm{for} \quad x \in\partial\Omega , 
\end{equation}
where we denote the boundary of $\Omega$ by $\partial\Omega$, and the initial condition
\begin{align}\label{Eq:ICs}
p(0,x) = p^0(x)\geq0, \quad \text{with} \quad \int_\Omega p^0(x) \d x>0,
\end{align}
where $p^0(x)$ denotes the phenotypic density function at time zero.

The second term on the left-hand-side of Eq.~\eqref{Eq:StructuredPDE} models spontaneous phenotypic changes as a diffusive flux~\citep{chisholm2016evolutionary,chisholm2016cell,cho2018modeling} with constant diffusion coefficient $\beta\geq0$.  The third term on the left-hand-side of Eq.~\eqref{Eq:StructuredPDE} models environment-driven phenotypic changes by an advection term~\citep{almeida2024evolutionary,celora2021phenotypic} with velocity $V(t,x)$, the time-dependency of which is likely mediated by some environmental factor  denoted by $c(t)\geq0$, i.e. $V(t,x)\equiv V(c(t),x)$. 
The reaction term on the right-hand-side of Eq.~\eqref{Eq:StructuredPDE} models phenotype-dependent cell proliferation and death as in the non-local Lotka-Volterra equation~\citep{ perthame2008dirac}. The phenotype-dependent intrinsic growth rate $f(t,x)$ is likely mediated by some environmental factor $c(t)$, i.e. $f(t,x)\equiv f(c(t),x)$, while the rate of death due to competition for space depends on the population size $P(t)$, defined in~\eqref{Eq:PopulationSize}, and the constant carrying capacity coefficient $\kappa>0$. 

We assume that the functions $f(t,x)$ and $V(t,x)$ are continuous in $x$ at each point in time, i.e.
\begin{equation}\label{ass:fv:continuous}
    f(t,\cdot) \in {C}^0(\Omega) \quad \text{and} \quad V(t,\cdot) \in {C}^0(\Omega) \quad \forall\, t\in\mathbb{R}_{\geq0}\,,
\end{equation}
and bounded in $t$ for each phenotypic state, i.e.
\begin{equation}\label{ass:fv:bounded}
    f_m \leq f(\cdot , x)  \leq f_M \quad \text{and} \quad V_m \leq V(\cdot , x)  \leq V_M \quad \forall x\in\Omega\subset \mathbb{R}\,.
\end{equation}

Models comprising PDEs in the form of Eq.~\eqref{Eq:StructuredPDE} can be formally derived from stochastic individual based models in the continuum, deterministic limit, see for instance~\citep{champagnat2002canonical,champagnat2006unifying,chisholm2016evolutionary,stace2020discrete} and references therein. In particular, the diffusion and drift terms emerge as the macroscopic deterministic description of a biased random walk \citep{chisholm2016evolutionary,lorenzi2020discrete,stace2020discrete}. In the stochastic and statistical literature, the left-hand-side of Eq.~\eqref{Eq:StructuredPDE} is usually thought of as a Fokker-Planck equation (or equivalently, a Kolmogorov forward equation) \citep{Kadanoff.2000}, which arises as the governing equation for the probability density function of a set of non-interacting particles undergoing a biased diffusion process in $x$.

\begin{remark}\label{Remark:SpatialInterpretation}
    The phenotypic state $x$ of a cell can be interpreted directly as the cellular level of expression of some gene and/or protein which mediates the observable characteristics and behaviour of the cell, relevant to the specific problem of interest. Due to natural biological constraints, gene and protein expression levels live in a bounded domain, as already clarified in the Introduction, motivating the choice of $\Omega:=[l,L]\subset\mathbb{R}$. The value of $l$ and $L$, i.e. the lowest and highest gene/protein expression levels realistically admissible, should be carefully selected by the modeller and inferred from biological data. In practice, these bounds encompass the entirety of the observable data. Consequently, gene/protein expression levels outside $\Omega$ are expected to be biologically infeasible and can be neglected.
\end{remark}

\section{Reduction to a system of ODEs characterising the phenotypic distribution}\label{Sec:Analysis}

The structured PDE~\eqref{Eq:StructuredPDE} captures the dynamics of the density of cells in phenotype-space. However, the density of cells with a given phenotype is unlikely to be the object of experimental or clinical interest. Rather, the evolution of the population size and distribution in phenotype space is relevant for understanding phenotypic adaptation. Consequently, we now generate a system of ODEs to characterise the population distribution in phenotype space. We begin by considering the size of the total population, $P(t).$

\subsection{The total cell population}

The population size $P(t)$ only depends on time due to the integration over phenotype space. We can therefore derive an integro-differential equation for the population size $P(t)$. 

\begin{lemma}\label{lemma:m0}
Let $p(t,x)$ satisfy Eq.~\eqref{Eq:StructuredPDE}, with $f$ satisfying assumption~\eqref{ass:fv:bounded}, along with boundary conditions~\eqref{Eq:BCs}, initial conditions~\eqref{Eq:ICs} and definition~\eqref{Eq:PopulationSize}. The population size $P(t)$ evolves according to the integro-differential equation
\begin{align}\label{Eq:PopulationSizeODE}
 \TimeDeriv  P(t) =  \int_\Omega f(t,x) p(t,x) \d x - \frac{P^2(t)}{\kappa} ,
\end{align}
complemented with the initial condition
\begin{equation}\label{ic:m0}
    P(0) = \int_\Omega p^0(x) \d x>0.
\end{equation}
Moreover, under the assumption in Eq~\eqref{ass:fv:continuous}, we have that
\begin{equation}\label{m0:bound}
    0< P(t)\leq \overline{P}<\infty \quad \forall t\geq0.
\end{equation}
\end{lemma}
The proof follows standard calculations, cf. Appendix~\ref{app:proof:prop}.

\subsection{The moment generating function}
\label{sec:mgf}

In Lemma~\ref{lemma:m0}, we derived an integro-differential equation for the population size $P(t)$. However, this integro-differential equation explicitly depends on the phenotypic density function. Rather than studying this explicitly, we instead characterise $p(t,x)$ by recasting it as a probability distribution in phenotype space and studying the moments of this distribution. 
Hence, in what follows, we consider the phenotypic density function scaled by the total population size
\begin{equation}\label{def:phat}
    \phat(t,x) = \frac{p(t,x)}{P(t)}.
\end{equation}
The distribution $\phat(t,x)$ encodes a time-dependent probability measure $\mu(t)$ over phenotype space. This measure $\mu(t)$ has a Radon-Nikoydym derivative with respect to the Lebesgue measure $\lambda$ given by the distribution $\phat(t,x)$, i.e.
\begin{align*}
\frac{ \d \mu}{\d \lambda} =  \phat(t,x). 
\end{align*} This measure $\mu(t)$ and the population size $P(t)$ is sufficient to describe the phenotype-structured population $p(t,x)$. In what follows, we develop a system of ODEs to characterise the moments of the distribution $\phat(t,x)$; \citet{Curto2023} performed a similar analysis for the heat equation. We consider the moment generating function of the distribution $\phat(t,x)$, given by
\begin{equation}\label{def:mgf}
    M(s,t) = \int_{\Omega}  e^{sx} \phat(t,x) \d x. 
\end{equation}
We see from this definition that $M(0,t) = 1$ due to the scaling of $\phat(t,x)$ by the total population size at all times $t$. The higher moments of $\phat(t,x)$, where $m_k(t)$ denotes the $k$-th moment, are given by
\begin{equation}\label{def:hom}
m_k(t) = \partial_s^{k} M(s,t)|_{s=0} \qquad  k\geq1.
\end{equation}
Similar to  \citet{Curto2023}, these higher moments are explicitly time dependent. As $\Omega$ is compact, it follows from the Stone-Weierstrass theorem and the solution of the Hausdorff Moment Problem that the sequence of moments, $\{m_k(t)\}_{k=1}^{\infty}$, uniquely determines the distribution $\phat(t,x)$. Indeed, if $\Omega$ is not compact, as in some of our examples, then the mapping between moments and distribution is more subtle. 

Using the definition~\eqref{def:phat}, the ODE for the evolution of the population size $P(t)$ in Eq.~\eqref{Eq:PopulationSizeODE} becomes
\begin{align}\label{Eq:PopulationSizeODE2}
 \TimeDeriv  P(t) =  \left(\int_\Omega f(t,x) \phat(t,x) \d x - \frac{P(t)}{\kappa} \right) P(t).
\end{align}
We note that $P(t)$ thus satisfies a generalized logistic equation with growth rate and carrying capacity dependent on the phenotypic distribution $\phat(t,x)$. Thus, we now focus on the evolution of $\phat(t,x)$. 

\subsection{A system of integro-differential equations for the moments of $\phat(t,x)$}

\begin{proposition}\label{Prop1} 
Let $p(t,x)$ satisfy Eq.~\eqref{Eq:StructuredPDE}, along with boundary conditions~\eqref{Eq:BCs}, initial conditions~\eqref{Eq:ICs} and definition~\eqref{Eq:PopulationSize}. 
Then, the 0-th moment of $\phat(t,x)$ defined in~\eqref{def:phat} 
is $m_0(t) =1$ for all $t \geq 0$ and, under assumption~\eqref{ass:fv:continuous}, the moments $m_k(t)$ ($k\in\mathbb{N}$, $k\geq1$) satisfy the following system of integro-differential equations
\begin{equation} \label{Eq:GenericMomentIDE}
\left. 
\begin{aligned}
\TimeDeriv m_1(t)  = & \; \left[ -\beta \left[ \phat(t,x) \right]\rvert_{\partial\Omega} +  \int_{\Omega} V(t,x) \phat(t,x) \d x   + \int_{\Omega} x f(t,x) \phat(t,x) \d x  -   m_1(t) \int_{\Omega} f(t,x) \phat(t,x)  \d x \right], \\
\TimeDeriv m_k(t)  = & \;\beta k(k-1) m_{k-2}(t) - \beta n\left[ x^{k-1} e^{xs}\phat(t,x) \right]\rvert_{\partial\Omega} +  k \int_{\Omega} x^{k-1} V(x,c)\phat(t,x) \d x \\
&  + \int_{\Omega} x^{k} f(x,c) \phat(t,x) \d x - m_k(t)   \int_{\Omega} f(x,c)\phat(t,x) \d x \qquad k\geq2,
\end{aligned}
\right \}
\end{equation}
complemented with initial conditions
\begin{equation}\label{Eq:GenericMomentIDE:ICs}
    m_k(0)= \left( \frac{1}{\int_{\Omega} p^0(s) \d s} \right) \int_\Omega x^k p^0(x) \d x \quad (k\geq1) 
\end{equation}
and the identity $m_0(t) = 1$ for all $t \geq 0.$
\end{proposition}

\begin{proof}
The proof of Proposition~\ref{Prop1} relies on the use of the moment generating function  of the distribution, introduced in Eq.~\eqref{def:mgf} of Section~\ref{sec:mgf}, to derive the higher order moments. 

\underline{\textit{Step 1: 0-th moment.}}
It follows immediately from the definition of $\phat(t,x)$ in Eq.~\eqref{def:phat} and the moment generating function in Eq.~\eqref{def:mgf} that $m_0(t) = 1$ for all time. 

\underline{\textit{Step 2: {evolution of the moment generating function.}}} 
We multiply \eqref{Eq:StructuredPDE} by $e^{sx}$ and integrate with respect to $x$ to find 
\begin{equation*}
    \partial_t \left[ \int_{\Omega} e^{sx} p(t,x) \d x \right] - \int_{\Omega} \left( e^{sx} \partial_x \left[ \beta \partial_x p(t,x) - V(t,x)p(t,x)\right] \right) \d x = \int_{\Omega} \left( f(t,x)- \frac{P(t)}{\kappa} \right) e^{sx} p(t,x) \d x . 
\end{equation*}

Multiplying the first term by unity, i.e. by $P(t)/P(t)$, and using definition~\eqref{def:mgf} yields
\begin{equation*}
    \partial_t \left[ M(s,t)P(t) \right] -\int_{\Omega} \left( e^{sx} \partial_x \left[ \beta \partial_x p(t,x) - V(t,x)p(t,x)\right] \right) \d x = \int_{\Omega} \left( f(t,x)- \frac{P(t)}{\kappa} \right) e^{sx} p(t,x) \d x .
\end{equation*}
The second term on the left-hand-side can be integrated by parts twice and, after imposing boundary condition~\eqref{Eq:BCs}, this gives
\begin{equation*}
\int_{\Omega}   e^{sx}  \partial_x \left[ \beta \partial_x p(t,x) - V(t,x)p(t,x)\right]  \d x  = -s \beta { \left[ e^{xs}p(t,x) \right]\rvert_{\partial\Omega} }
+ \beta s^2 M(s,t) P(t)  + s \int_{\Omega} e^{sx} V(t,x) p(t,x) \d x. 
\end{equation*}
Altogether this gives
\begin{align*}
 P(t) \partial_t M(s,t)  +  M(s,t) \TimeDeriv P(t)  =  & -s \beta  { \left[ e^{xs}p(t,x) \right]\rvert_{\partial\Omega} } + \beta s^2 M(s,t) P(t) + s \int_{\Omega} e^{sx} V(x,c)p(t,x) \d x \\
 &   + \int_{\Omega} f(x,c)e^{sx} p(t,x) \d x   - \frac{P(t)}{\kappa}  M(s,t)m_0(t).
\end{align*}
Substituting \eqref{Eq:PopulationSizeODE} and diving by $P(t) > 0$, which is non-zero as proved in Lemma~\ref{lemma:m0}, we find
\begin{equation} \label{Eq:MGFTimeDerivative}
\left. 
\begin{aligned}
\partial_t M(s,t)   = & - s \beta  { \left[ e^{xs}\phat(t,x) \right]\rvert_{\partial\Omega} } + \beta s^2 M(s,t) + s \int_{\Omega} e^{sx} V(t,x)\phat(t,x) \d x \\  
 &     + \int_{\Omega} f(t,x)e^{sx} \phat(t,x) \d x   -  M(s,t) \int_{\Omega} f(t,x) \phat(t,x)  \d x ,
\end{aligned}
\right \}
\end{equation}
where we used definition~\eqref{def:phat}. 

\underline{\textit{Step 3: the first moment.}} 
Differentiating \eqref{Eq:MGFTimeDerivative} once with respect to $s$ gives
\begin{align*}
\partial_t[\partial_s M(s,t)]  = &  -  \beta  \big( { \left[ e^{xs}\phat(t,x) \right]\rvert_{\partial\Omega} } 
  + s  { \left[ x e^{xs}\phat(t,x) \right]\rvert_{\partial\Omega} } \big) + \beta  \left( 2s M(s,t) + s^2 \partial_s M(s,t)\right) {+ \int_{\Omega} e^{sx} V(t,x)\phat(t,x) \d x} \\
& {+ s \int_{\Omega} x e^{sx} V(t,x)\phat(t,x) \d x} + \int_{\Omega} x f(t,x)e^{sx} \phat(t,x) \d x   -  [\partial_s M(s,t)] \int_{\Omega} f(t,x) \phat(t,x)  \d x .
\end{align*}
which, after setting $s =0$, immediately gives 
\begin{align*}
\TimeDeriv m_1(t) = - \beta \left[ \phat(t,x) \right]\rvert_{\partial\Omega} {+ \int_{\Omega} V(t,x)\phat(t,x) \d x}
 + \int_{\Omega} x f(t,x) \phat(t,x) \d x  -  m_1(t) \int_{\Omega} f(t,x) \phat(t,x)  \d x,
\end{align*}
which is the first equation in Eq.~\eqref{Eq:GenericMomentIDE}. This is complemented with the initial condition 
\begin{equation*}
    m_1(0)= \int_\Omega x \phat(0,x) \d x = \frac{1}{P(0)}\int_\Omega x p^0(x) \d x ,
\end{equation*}
obtained from the definition~\eqref{def:hom} and initial condition~\eqref{Eq:ICs}. 

\underline{\textit{Step 4: the $k$-th moment.}} For $k \geq 2$, we calculate the $k$-th derivatives with respect to $s$ of the terms on the right-hand-side of \eqref{Eq:MGFTimeDerivative} to find{, by induction,} the following:
\begin{align*}
\partial^k_s \big[  s \beta   \left[ e^{xs}\phat(t,x) \right]\rvert_{\partial\Omega} \big] &= \beta \big( k{ \left[ x^{k-1} e^{xs}\phat(t,x) \right]\rvert_{\partial\Omega} } 
  + s  { \left[ x^k e^{xs}\phat(t,x) \right]\rvert_{\partial\Omega} } \big),  \\
 \partial^k_s \left[  \beta s^2 M(s,t) \right]  &= \beta k(k-1) \partial_s^{k-2} M(s,t) +\beta \left( 2k s\partial_s^{k-1} M(s,t)  +  s^2 \partial_s^{k} M(s,t) \right) ,  \\
\partial^k_s \left[  s \int_{\Omega} e^{sx} V(x,c)\phat(t,x) \d x \right] &= \int_{\Omega} \left( kx^{k-1} + s  x^{k} \right) e^{sx} V(x,c)\phat(t,x) \d x , \\
\partial^k_s \left[  \int_{\Omega} e^{sx} f(x,c)\phat(t,x) \d x \right] &= \int_{\Omega} x^{k} e^{sx} f(x,c)\phat(t,x) \d x.
\end{align*}
Then, differentiating~\eqref{Eq:MGFTimeDerivative} $k$ times and setting $s=0$, one retrieves, for $k \geq 2$, 
 \begin{equation} \label{Eq:HigherOrderMomentODE}
 \left. 
 \begin{aligned}
 \TimeDeriv m_k(t)  =  &  - \beta k \left[ x^{k-1} e^{xs}\phat(t,x) \right]\rvert_{\partial\Omega} + \beta  k(k-1) m_{k-2}(t) +  \int_{\Omega} nx^{k-1} V(x,c)\phat(t,x) \d x  \\
 &   + \int_{\Omega} x^{k} f(x,c) \phat(t,x) \d x - m_k(t)  \int_{\Omega} f(x,c)\phat(t,x) \d x,
 \end{aligned}
 \right \}
 \end{equation}
 as in Eq.~\eqref{Eq:GenericMomentIDE}. This is complemented with the initial condition 
\begin{equation*}
    m_k(0)=\int_\Omega x^k \phat(0,x) \d x = \frac{1}{P(0)}\int_\Omega x^k p^0(x) \d x ,
\end{equation*}
obtained, again, from the definition~\eqref{def:hom} and initial condition~\eqref{Eq:ICs}, completing~\eqref{Eq:GenericMomentIDE:ICs}. 

\end{proof}

 \begin{remark} \label{Remark:BoundaryZerop}
 The integro-differential equations for $m_k(t)$ include the distribution $\phat(t,x)$ evaluated on the boundary $\partial \Omega.$ Due to diffusion, the distribution is not identically zero at the boundary. However, the biological interpretation of the phenotypic variable (see Remark~\ref{Remark:SpatialInterpretation}) implies that the population density at the boundary, while non-zero, is sufficiently small to be unobservable in biological data. Therefore, in what follows, we make the biologically motivated assumption that the contribution of the boundary terms is negligible, so
 \begin{align}\label{Eq:BoundaryAssumption}
       \left.  \left[ \phat(t,x) \right] \right| _{\partial\Omega} =  
 \left. \left[x^k \phat(t,x) \right]\ \right| _{\partial\Omega} = 0.
 \end{align}
     
 \end{remark}

\subsection{Restriction to a bounded phenotypic domain} \label{Sec:TruncationSection1} 

The system of equations in \eqref{Eq:GenericMomentIDE} involves integrating $\phat(t,x)$ over the entire phenotypic domain $\Omega$. Consequently, without making any further assumptions on the phenotypic distribution, the system~\eqref{Eq:GenericMomentIDE} is circular as the resulting integro-differential equations for the moments $m_n(t)$ \textit{a priori} require the distribution $\phat(t,x)$. Until now, we have considered a generic compact and connected phenotypic domain $\Omega=[l,L]$,  intrinsic growth rate $f(t,x)$, and adaptation velocity 
$V(t,x)$. Importantly, the integral terms in \eqref{Eq:GenericMomentIDE} depend on these functions and implicitly on the domain $\Omega$. Indeed, it is possible to restrict $\Omega$ to the unit interval via a simple linear transformation~\citep{almeida2024evolutionary,stace2020discrete}.

We now show that, by considering a bounded phenotypic domain restricted to $\Omega = [0,1]$, and functions $f(t,x)$ and $V(t,x)$ that are analytic in phenotype space, we can eliminate the redundancy in Eq.~\eqref{Eq:GenericMomentIDE}. Specifically, building on the analytical strategies adopted in~\citep{Dieckmann1996,Engblom:2006,lee2009moment}, we show how utilizing the Taylor expansions of both $f(t,x)$ and $V(t,x)$ transforms Eq.~\eqref{Eq:GenericMomentIDE} into a system of differential equations that depend only on the moments $m_k(t)$, $k\geq1$. 
\begin{corollary}
Consider $x\in\Omega\equiv [0,1]$. Let $p(t,x)$ satisfy Eq.~\eqref{Eq:StructuredPDE}, along with boundary conditions~\eqref{Eq:BCs}, initial conditions~\eqref{Eq:ICs} and definition~\eqref{Eq:PopulationSize}. Further, assume that both $f$ and $V$ are analytic functions of $x$ and that equality~\eqref{Eq:BoundaryAssumption} holds. 

Then, $m_0(t) =1$ for all $t \geq 0$ by definition, and the higher order moments $m_k(t)$ ($k \geq 1$) of the phenotypic distribution $\phat(t,x)$ satisfy the following system of ODEs
\begin{equation} \label{Eq:ExpandedGenericMomentODE}
\left. 
\begin{aligned}
\TimeDeriv m_1(t) & =  \displaystyle \sum_{n=0}^{\infty} \left(V_n(t)  - m_1(t)f_n(t) \right) \left[ \displaystyle \sum_{i=0}^n (-1)^{i} \binom{n}{i} (m_1(t))^{i}  \right] m_{n-i}(t)  \\ 
& {} \quad + \displaystyle \sum_{n=0}^{\infty}  f_n(t) \left[  \displaystyle \sum_{i=0}^n  (-1)^{i} \binom{n}{i} (m_1(t))^{i} m_{n+1-i}(t) \right]  \\
 \TimeDeriv m_k(t) & = - m_k(t) \displaystyle \sum_{n=0}^{\infty} f_n(t)  \displaystyle \left[ \sum_{i=0}^n  (-1)^{i}  \binom{n}{i} (m_1(t))^{i} m_{n-i}(t) \right]   + \beta  k(k-1) m_{k-2}(t) \\ 
    & {} \quad   + \displaystyle \sum_{n=0}^{\infty} \left[  \displaystyle \sum_{i=0}^n  (-1)^{i} \binom{n}{i} (m_1(t))^{i}  \right] \left[ f_n(t)m_{n+k-i}(t)  + kV_n(t)  m_{n+(k-1)-i}(t) \right]  , \quad k \geq 2.
 \end{aligned}
\right \}
\end{equation}
with initial conditions given by Eq.~\eqref{Eq:GenericMomentIDE:ICs}, and where $f_n(t)$ and $V_n(t)$ are defined as
\begin{align}\label{def:fnvn}
f_n(t) = \frac{\partial_x^n f(t,x)|_{x=m_1(t)}}{n!} \quad \textrm{and} \quad   V_n(t) = \frac{\partial_x^n V(t,x)|_{x=m_1(t)}}{n!}.
\end{align} 
\end{corollary}

\begin{proof}
    As both $f$ and $V$ are analytic functions of $x$, we Taylor expand these functions about the 
    first moment $m_1(t)$ to find
\begin{align*}
f(x,t) = \displaystyle \sum_{n=0}^{\infty} \frac{\partial_x^n f(t,x)|_{x=m_1(t)}(x-m_1(t))^n}{n!} \quad \textrm{and} \quad V(x,t) = \displaystyle \sum_{n=0}^{\infty} \frac{\partial_x^n V(x,t)|_{x=m_1(t)}(x-m_1(t))^n}{n!}.
\end{align*}
    Then, using definitions~\eqref{def:fnvn}, the binomial expansion of $(x-m_1(t))^n$ gives
\begin{align*}
f(x,c) = \displaystyle \sum_{n=0}^{\infty}  f_n(t) \left[ \displaystyle \sum_{i=0}^n (-1)^{i} \binom{n}{i} x^{n-i}(m_1(t))^{i} \right] \quad \textrm{and} \quad V(x,c) = \displaystyle \sum_{n=0}^{\infty} V_n(t)  \left[ \displaystyle \sum_{i=0}^n (-1)^{i} \binom{n}{i} x^{n-i} (m_1(t))^{i} \right].
\end{align*}
Thus, utilizing the Taylor expansion of $f$ and the definition of $m_n$, combining~\eqref{def:mgf}-\eqref{def:hom}, gives
\begin{align*}
\int_0^1 f(x,c) \phat(t,x) \d x & = \displaystyle \sum_{n=0}^{\infty} \left[ \displaystyle \sum_{i=0}^n  (-1)^{i} f_n(t) \binom{n}{i} (m_1(t))^{i} \int_0^1 x^{n-i}\phat(t,x) \d x \right] \\
& = \displaystyle \sum_{n=0}^{\infty}   \displaystyle \left[ \sum_{i=0}^n  (-1)^{i} f_n(t) \binom{n}{i} (m_1(t))^{i} m_{n-i}(t) \right] .
\end{align*}
Then, for integer $k$, a similar calculation yields
\begin{align*}
\int_0^1 x^{k} f(x,c)\phat(t,x) \d x = \displaystyle \sum_{n=0}^{\infty} \left[  \displaystyle \sum_{i=0}^n  (-1)^{i} f_n(t) \binom{n}{i} (m_1(t))^{i} m_{n+k-i}(t) \right],
\end{align*} 
and 
\begin{align*}
\int_0^1 k x^{k-1} V(x,c) \phat(t,x) \d x = \displaystyle \sum_{n=0}^{\infty} V_n(t)  \left[ \displaystyle \sum_{i=0}^n k (-1)^{i} \binom{n}{i} (m_1(t))^{i}  m_{n+(k-1)-i}(t)\right] .  
\end{align*} 
Inserting these expansions into the ODEs~\eqref{Eq:GenericMomentIDE}, and using~\eqref{Eq:BoundaryAssumption}, immediately yields Eq.~\eqref{Eq:ExpandedGenericMomentODE}. 
\end{proof}

These Taylor expansions replace the integral terms in \eqref{Eq:GenericMomentIDE} by weighed moments of the distribution $\phat(x,t)$. However, the resulting differential equations involve moments of all orders, each of which needs to be defined by a corresponding ODE. Consequently, replacing the integral terms by the corresponding infinite summations in \eqref{Eq:GenericMomentIDE} leads to a system of infinitely many ODEs wherein the differential equation for the $k$-th moment depends on higher order moments.  Nevertheless, higher order moments are not generally used to describe biological data. Therefore, we proceed under the modelling assumption that the phenotypic density function $p(t,x)$, and thus $\phat$, is sufficiently well characterised by its first $N$ moments, so we assume that we can discard the higher order moments.  
We illustrate how this assumption has been applied in existing models in Section~\ref{Sec:Examples} and discuss the limitations of this assumption in the Discussion.  
Moreover, we consider the asymptotic behaviour of the terms in the summations,  to truncate the series and close the system.

\subsection{Series truncation and system closure} \label{Sec:TruncationSection2} 
Consider the infinite series appearing in Eq.~\eqref{Eq:ExpandedGenericMomentODE}$_1$, i.e. the term 
$$
T(t) = \sum_{n=0}^{\infty} f_n(t)\left[\displaystyle \sum_{i=0}^n   (-1)^{i}  \binom{n}{i} (m_1(t))^{i} m_{n-i}(t)\right] \,.
$$

Consider the fact that $0<m_1(t)<1$ for all $t\geq0$, where the strict inequality can be ensured by appropriate modelling choices, as discussed in Remark~\ref{Remark:BoundaryZerop}.  The strict upper bound on $m_1(t)$ implies that $(m_1)^i\to0$ as $i\to\infty$.  Having chosen $f(t,x)$ analytic in a bounded domain, $f_n(t)$ is bounded for all $n\in\mathbb{N}$, so the coefficients of $(m_1)^i$ are bounded. We may choose $M\in\mathbb{N}$ at which to truncate the Taylor series expansions for $f$ and $V$.
This allows us to truncate the infinite summations at $M$ and approximate $T(t)$ by 
\begin{align*}
T_{M}(t) =\sum_{n=0}^{M} f_n(t) \left[ \displaystyle \sum_{i=0}^n  (-1)^{i}\binom{n}{i} (m_1(t))^{i} m_{n-i}(t) \right]\,.
\end{align*}
The same truncation can be applied to the infinite summations in Eq~\eqref{Eq:ExpandedGenericMomentODE} involving $V_n$. 

Moreover, as a result of the assumption that the phenotypic distribution $\phat(t,x)$ is sufficiently well characterised by its first $N$ moments, we find that the system~\eqref{Eq:ExpandedGenericMomentODE} can be approximated by a finite system of ODEs for the $N$ moments of $\phat(t,x)$. 
The overall dynamics of the phenotypic density function of the population are approximated by
\begin{equation}\label{Eq:TruncatedGenericMomentODE:open}
\left. 
\begin{aligned} 
\TimeDeriv  P(t) & =  P(t) \displaystyle \sum_{n=0}^{M} f_n(t) \left[ \displaystyle \sum_{i=0}^n  (-1)^{i}  \binom{n}{i} (m_1(t))^{i} m_{n-i}(t) \right] - \frac{P^2(t)}{\kappa} , \\
\TimeDeriv m_1(t) & =  \displaystyle \sum_{n=0}^{M} \left(V_n(t)  - m_1(t)f_n(t) \right) \left[ \displaystyle \sum_{i= 0}^n (-1)^{i} \binom{n}{i} (m_1(t))^{i}  m_{n-i}(t) \right] \\
&   \quad + \displaystyle \sum_{n=0}^{M} f_n(t) \left[  \displaystyle \sum_{i=0}^n  (-1)^{i} \binom{n}{i} (m_1(t))^{i} m_{n+1-i}(t) \right] , \\
\TimeDeriv m_k(t) & = - m_k(t) \displaystyle \sum_{n=0}^{M} f_n(t)  \displaystyle \left[ \sum_{i= 0}^n  (-1)^{i}  \binom{n}{i} (m_1(t))^{i} m_{n-i}(t) \right]   + \beta  k(k-1) m_{k-2}(t)  \\
    &   \quad + k \displaystyle \sum_{n=0}^{M} V_n(t)  \left[ \displaystyle \sum_{i= 0}^n  (-1)^{i} \binom{n}{i} (m_1(t))^{i} m_{n+(k-1)-i}(t)  \right]    \\
    & \quad  + \displaystyle \sum_{n=0}^{M} f_n(t)  \left[  \displaystyle \sum_{i= 0}^n  (-1)^{i} \binom{n}{i} (m_1(t))^{i} m_{n+k-i}(t) \right] , \quad 2 \leq k \leq N .   
\end{aligned}
\right \}
\end{equation}
The corresponding initial conditions of Eq.~\eqref{Eq:TruncatedGenericMomentODE:open} are obtained directly from the initial phenotypic distribution $\phat(0,x)$ via Eq.~\eqref{Eq:GenericMomentIDE:ICs}. 
While system~\eqref{Eq:TruncatedGenericMomentODE:open} is finite, it still involves the moments $m_k$ for $k=N+1,...,N+M$. To close the system one may choose, for instance, Gaussian closure by setting
\begin{equation}\label{closure:Gaussian}
    m_k = m_1m_{k-1} + k(m_2-m_1^2)m_{k-2} \quad \text{for} \quad k=N+1,...,N+M,
\end{equation}
or truncation closure by setting
\begin{equation}\label{closure:truncation}
    m_k = 0 \quad \text{for} \quad k=N+1,...,N+M,
\end{equation}
-- in which case system~\eqref{Eq:TruncatedGenericMomentODE:open} reads as~\eqref{Eq:TruncatedGenericMomentODE}. 

Other approaches to closing this system of ODEs to approximate the moments of the phenotypic distribution are possible \cite{Browning.2020,wagner2022quasi}.

\section{Special cases and examples in the literature}\label{Sec:Examples}

We here expand on some examples of how our generalised approach can be applied to more specific cases. 
We recover ODE systems for the moments of the phenotypic distribution previously considered in the literature from the generalised system~\eqref{Eq:ExpandedGenericMomentODE}, derived in Section~\ref{Sec:Analysis}. 
This procedure also brings to light the details of how the obtained systems depend on the underlying assumptions on the nature of the phenotypic distribution or specific modelling choices. 

\subsection{The case of $f$ and $V$ polynomials}\label{sec:example:polynomial}

We now focus on the case in which $f(t,x)$ and $V(t,x)$ are defined as polynomial functions of $x$, as done in various mathematical models employing PDEs in the form of Eq.~\eqref{Eq:StructuredPDE} for the adaptive dynamics of a phenotype-structured population of cells, e.g. see~\cite{almeida2024evolutionary,chisholm2016evolutionary} and references therein.

If the functions $f$ and $V$ are polynomials then, as in~\cite{Curto2023}, one need not restrict the domain to a bounded set for the results of Section~\ref{Sec:Analysis} to hold.  
In fact, a polynomial of order $D\in\mathbb{N}$ can easily be expressed in the form of a Taylor series truncated at $D$, as all derivatives of order higher than $D$ will be zero.  It is then natural to replace $M$ in the upper bounds of the summations in system~\eqref{Eq:TruncatedGenericMomentODE} by $D:=\max (D_f,D_V)$, where  $f(t,x)$ and $V(t,x)$ are polynomials of order $D_f$ and $D_V$, respectively.  
Then one need not rely on the restriction of $\Omega$ to the interval $[0,1]$ to ensure that the infinite summations, a product of the Taylor expansion of $f$ and $V$, can be truncated at some $M\in\mathbb{N}$ which, we re-iterate, rely on the fact that $0<m_1<1$ under this domain restriction.  
This allows the extension of the results to the case in which $\Omega=\mathbb{R}$, as considered in previous works deriving a system of ODEs for the moments of the phenotypic density function starting from phenotype-structured PDEs~\citep{almeida2024evolutionary,almeida2019evolution,ardavseva2020evolutionary,chisholm2016evolutionary,lorenzi2015dissecting,villa2021evolutionary}. 

We remark that, while choosing $f$ and $V$ to be polynomials allows one to bypass choosing $M$ for truncation of the infinite series, it does not automatically close the system of ODEs for the moments, and one is still required to identify the highest moment required to characterise the distribution to achieve this. Further, as $\Omega=\mathbb{R}$ is not compact, the moment sequence is not sufficient to characterise the distribution without additional assumptions. In the aforementioned papers, this was done implicitly by introducing stronger assumptions on the shape of the phenotypic distribution, which could be introduced only thanks to the use of an infinite domain. Let us expand on this with an example.

\paragraph{Example from the literature: a Gaussian distribution.} Consider $\Omega\equiv\mathbb{R}$, as well as $f(t,x)$ and $V(t,x)$ defined by
\begin{equation}\label{Eq:def:fv}
    f(t,x) = a(t) - b(t) \left( x- X(t) \right)^2  \quad  \text{and} \quad V(t,x)\equiv V_0(t),
\end{equation}
i.e. polynomials of order 2 and 0, respectively. 
In definition~\eqref{Eq:def:fv} for $f$, $X(t)$ models the fittest phenotypic trait, $a(t)$ the associated maximal background fitness and $b(t)$ is a nonlinear selection gradient measuring the strength of the selective environment.
Under the assumption of an initial phenotypic distribution $p^0(x)$ in a Gaussian form, \cite{chisholm2016evolutionary} first showed that $p(t,x)$ maintains a Gaussian form at all times, with the population size $P(t)$, the mean $m_1(t)$ and variance $\sigma^2(t)$ of the distribution satisfying the system of ODEs
    \begin{equation}\label{Eq:GaussianODE}
\left. 
\begin{aligned} 
    \TimeDeriv P(t) =& \left[ a(t) - b(t)\left(m_1(t)-X(t)\right)^2 - b(t)\sigma^2(t) \right] P(t) - \frac{P^2(t)}{\kappa} , \\
    \TimeDeriv m_1(t) =& -2b(t)\left(m_1(t)-X(t)\right)\sigma^2(t) + V_0(t), \\
    \TimeDeriv \sigma^2(t) =& 2\beta -2b \sigma^4(t) , 
\end{aligned}
\right \}
\end{equation}
complemented with initial conditions $P(0)$, $m_1(0)$ and $\sigma^2(0)$, i.e. the corresponding moments characterising the initial Gaussian phenotypic density function $p^0(x)$. Analogous assumptions and results have since appeared in several following works by Lorenzi and coworkers~\citep{almeida2024evolutionary,almeida2019evolution,ardavseva2020evolutionary,lorenzi2015dissecting,villa2021evolutionary} modelling cancer adaptive dynamics in different settings. 
System~\eqref{Eq:GaussianODE} can be obtained via formal calculations following the substitution of a Gaussian ansatz in Eq.~\eqref{Eq:StructuredPDE} under definitions~\eqref{Eq:def:fv}, 
with most publications working with the inverse variance $v=(\sigma^{2})^{-1}$ for convenience. 


\paragraph{Retrieving system~\eqref{Eq:GaussianODE} from our generalised approach.} Now consider the case in which $f$ and $V$ are defined as in~\eqref{Eq:def:fv}. We show how the example above is a sub-case of our generalised approach, assuming that the phenotypic distribution can be fully characterised by its first 2 moments. We thus consider system~\eqref{Eq:TruncatedGenericMomentODE:open} with $N=2$, the infinite summations including $f_n(t)$ truncated at $D_f=2$ and those only including $V_n(t)$ truncated at $D_V=0$. We then choose to apply Gaussian closure -- indeed the obvious choice in this case -- and complement the ODE system with~\eqref{closure:Gaussian} for $m_3$ and $m_4$. The resulting system of ODEs is
\begin{equation*}\label{Eq:Ex1TruncatedGenericMomentODE}
\left. 
\begin{aligned} 
\TimeDeriv  P(t) & =  P(t) \left[ f_0(t) + f_2(t) (m_2 -m_1^2)\right] - \frac{P^2(t)}{\kappa} , \\
\TimeDeriv m_1(t) & = f_1(t) (m_2 -m_1^2) + V_0(t) , \\
\TimeDeriv m_2(t) & = 2V_0m_1 + 2\beta +2 f_1(t)m_1(m_2-m_1^2) +2f_2(t) (m_2 - m_1^2)^2,   
\end{aligned}
\right \}
\end{equation*}
which is equivalent to system~\eqref{Eq:GaussianODE} -- this is easy to check by substituting $m_2=\sigma^2+m_1^2$, from the definition of second central moment, and applying the definitions of $f_i$ ($i=0,1,2$) and $V_0$ in~\eqref{def:fnvn} on $f$ and $V$ chosen in~\eqref{Eq:def:fv}.

This demonstrates that system~\eqref{Eq:TruncatedGenericMomentODE:open} may provide a good approximation for the moment dynamics of Eq.~\eqref{Eq:StructuredPDE}. We stress that in this case the systems are equivalent only thanks to the fact that $f$ and $V$ are polynomials of order at most 2, and the properties~\eqref{closure:Gaussian} of a Normal distribution, which may be exploited when adopting a Gaussian closure. 
Unlike in~\citep{almeida2024evolutionary,almeida2019evolution,ardavseva2020evolutionary,chisholm2016evolutionary,lorenzi2015dissecting,villa2021evolutionary}, our procedure does not need to rely on an ansatz or an infinite domain, but it holds for the more realistic case of $\Omega=[l,L]$ and under more general assumptions on the nature of the phenotypic distribution. 

\subsection{An example for $f$ and $V$ smooth but not polynomial}\label{sec:example:new}

Indeed for $f$ and $V$ polynomials the truncation of the infinite summations in system~\eqref{Eq:TruncatedGenericMomentODE:open} is automatic. 
While polynomial definitions for $f$ and $V$ are widely used in continuously structured models in mathematical oncology, other modelling choices may lead to alternative definitions of these functions. In general, these functions are usually sufficiently smooth to admit a Taylor series approximation~\cite{celora2021phenotypic,Cho2018}. 
We now consider an example with $f$ and $V$ not polynomial, in which case they must be approximated by their Taylor expansions truncated at some $M$. Let $f$ and $V$ be defined as
\begin{equation}\label{ex:fV:adhoc}
    f(t,x) = f_{max}\frac{x}{k_x+x} \quad \text{and} \quad V(t,x) = V_{max}\tanh{(x^\omega)}\tanh{(1-x)}\,.
\end{equation}
This definition for $f$ may be used in the case where $x$ models the level of expression of some protein on the cell membrane that is required to transport nutrients into the cell, e.g. see~\cite{almeida2024evolutionary}, in which case a Hill function is natural following Michaelis-Menten kinetics. Then $k_x>0$ is the Michaelis coefficient and $f_{max}\geq0$ the maximum intrinsic growth rate. Definition~\eqref{ex:fV:adhoc} for $V$ was taken from~\cite{celora2021phenotypic}, one of the most complex choices of V currently found in the literature, proposed following phenomenological rules to capture the effect of radiotherapy on differentiation, with $V_{max}\geq0$ and $\omega\in\{1,2\}$. For simplicity, we take the initial phenotypic distribution to be a truncated Normal distribution in the interval $\Omega=[l,L]$, i.e. we take
\begin{equation}\label{ic:truncatednormal}
    p(0,x) = \displaystyle{P(0)\, \frac{\exp{\left[-\frac{(x-\bar{x}_0)^2}{2\sigma^2_0}\right]}}{\int_l^L \exp{\left[-\frac{(s-\bar{x}_0)^2}{2\sigma^2_0}\right] {\rm d}s}}}\,.
\end{equation}
The first step in the application of system~\eqref{Eq:TruncatedGenericMomentODE:open} to approximate the moment dynamics is to carefully select the domain bounds, to ensure that~\eqref{Eq:BoundaryAssumption} is satisfied. In~\cite{celora2021phenotypic} the authors select $x\in[0,1]$ but under definition~\eqref{ex:fV:adhoc} for $V$ we expect the mass to concentrate around $x=1$ over time and thus a wider domain is necessary. Here we use a parameter set for which the interval $[l,L]=[0,1.2]$ is sufficiently large -- with $\beta$ and $\sigma^2_0$ sufficiently small and $\bar{x}_0$ sufficiently far from these boundaries. We thus apply the transformation $x\to x/L$ prior to solving the ODE system numerically, to ensure that $0<m_1(t)<1$ for all $t\geq0$ as required in Section~\ref{Sec:TruncationSection2}. The results of numerical simulations shown in Fig~\ref{fig:ex2} (as well as Fig~\ref{fig:ex2:truncation} and~\ref{fig:ex2:om2} in the Supplemental Information) have been rescaled back to the original domain (by means of $m_k\to m_k\times L^k$). 
\begin{figure}[htb!]
    \centering
    \includegraphics[width=1\linewidth]{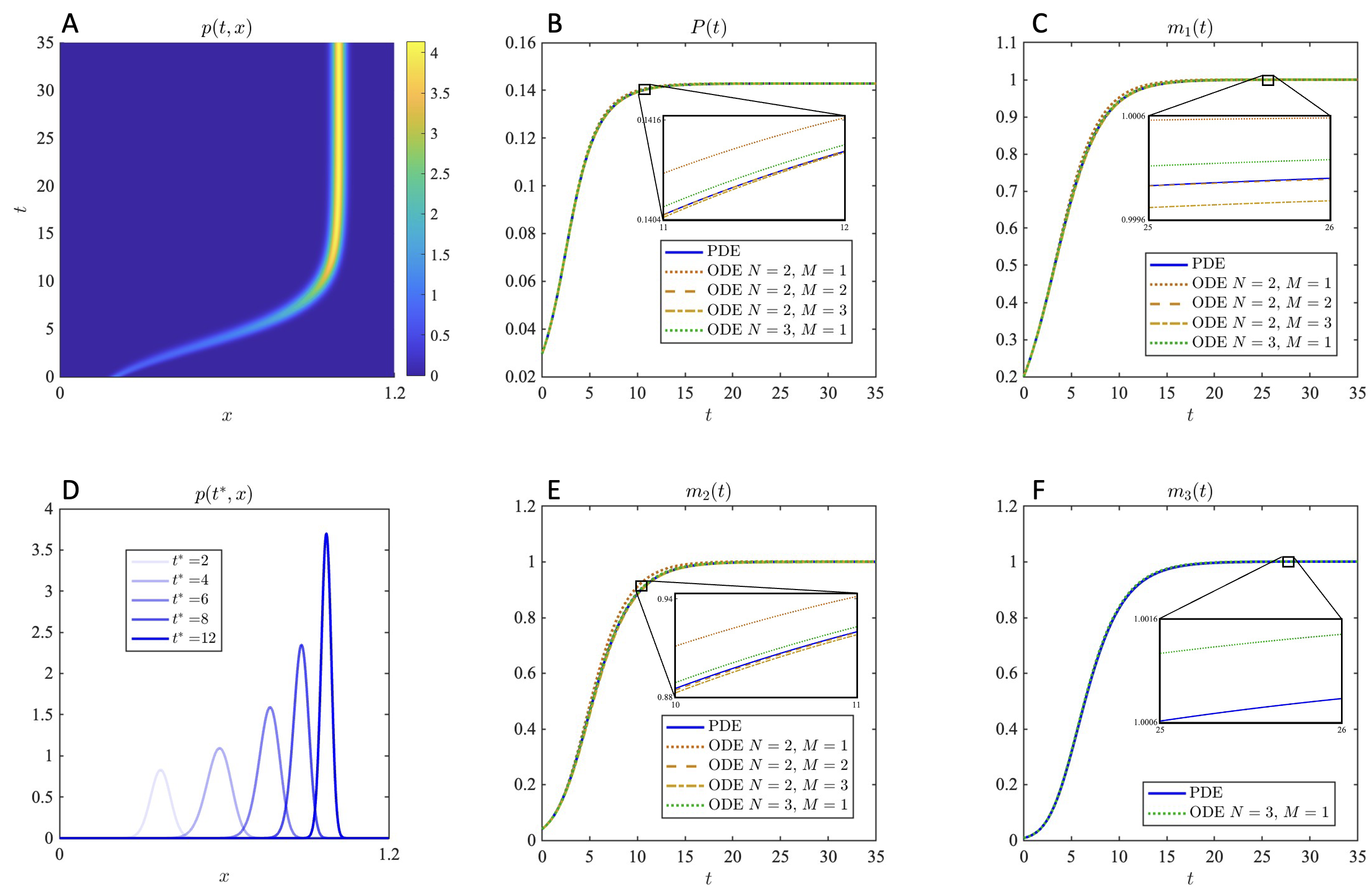}
    \caption{Comparison between the numerical solution of PDE~\eqref{Eq:StructuredPDE}, complemented with Eq.~\eqref{Eq:PopulationSize} and boundary conditions~\eqref{Eq:BCs}, and that of the ODE system~\eqref{Eq:TruncatedGenericMomentODE:open} under Gaussian closure~\eqref{closure:Gaussian}, for $f$ and $V$ defined in~\eqref{ex:fV:adhoc} ($\omega=1$) and initial condition~\eqref{ic:truncatednormal}. The phenotypic density function $p(t,x)$ obtained by solving the PDE is shown in the first column for $t\in[0,35]$ (panel A) and at selected times $t^*=2,4,6,8,12$ (D). The corresponding moments are plotted with a blue solid line in the remaining panels: the populations size $P(t)$ (B), the first moment $m_1(t)$ (C), the second moment $m_2(t)$ (E) and the third moment $m_3(t)$ (F). The moments obtained from the ODE system are also shown in each panel under different choices of $N$ and $M$: $N=2$ and $M=1$ (brown dotted line), $N=2$ and $M=2$ (brown dashed line), $N=2$ and $M=3$ (brown dash-dotted line), $N=3$ and $M=1$ (green dotted line). Zoomed-in insets are provided to facilitate comparison among approximations. 
    The results correspond to the parameter set: $f_{max}=2$, $k_m=0.4$, $V_{max}=0.5$, $\omega=1$, $\kappa=0.1$, $\beta=10^{-4}$, $\bar{x}_0=0.2$, $\sigma_0=0.02$, $l=0$, $L=1.2$. The numerical schemes rely on a first order forward difference approximation for the time derivatives, second order central difference for the diffusion term and first order upwind for the advection term.   }
    \label{fig:ex2}
\end{figure}

Setting $\omega=1$ in the definition~\eqref{ex:fV:adhoc} for $V$, we see in Fig~\ref{fig:ex2} that the ODE system~\eqref{Eq:TruncatedGenericMomentODE:open} under Gaussian closure~\eqref{closure:Gaussian} provides a really good approximation of the dynamics of the moments of the phenotypic density function satisfying PDE~\eqref{Eq:StructuredPDE} for many combinations of $N$ and $M$, without needing to consider $N$ and $M$ particularly large. In fact, it is interesting to notice from the zoomed in insets that the best approximation -- among those tested -- is obtained when selecting $N=2$ and $M=2$ rather than for higher values of $N$ and $M$ (c.f. $N=2$ and $M=3$, or $N=3$ and $M=1$). As expected when looking at the solution of the PDE (cf. first column in Fig~\ref{fig:ex2}), truncation closure does not perform as well as Gaussian closure (see also Fig~\ref{fig:ex2:truncation} in the Supplemental Information). In the more complex case of $\omega=2$, for which we observe a higher skewedness of the phenotypic distribution at intermediate times, not all choices of $N$ and $M$ provide as good an approximation (see also Fig~\ref{fig:ex2:om2} in the Supplemental Information). Nevertheless, the observation that $N=2$ and $M=2$ yields the best approximation among those tested holds also in this case, where we also tested $N=3$ and $M=3$. 
Ultimately, we expect the combination of $N$ and $M$ that best approximate the moment dynamics to vary with the definitions of $f$ and $V$, as well as the best choice of closure for system~\eqref{Eq:TruncatedGenericMomentODE:open}. Indeed the presence of a linear diffusion term in PDE~\eqref{Eq:StructuredPDE} may favour Gaussian-like features.

\subsection{Link with canonical model of adaptive dynamics} \label{Sec:PressleyExample}
We have seen how our model reduction procedure may perform well for non-polynomial, but sufficiently smooth, functions $f$ and $V$. Models of adaptive dynamics~\citep{Dieckmann1996,Vincent2008}, 
which model the size and mean trait of the population via two ODEs, typically utilize the ``$G$ function'' formalization \citep{Coggan2022}. This ``$G$ function'' links the population fitness to the phenotypic evolution by assuming that the population evolves to maximize the population fitness.
In general, these adaptive dynamics models do not restrict the phenotype space $\Omega$, but rather implicitly assume that the functions $f$ and $V$ are only weakly non-linear and that the population is monomorphic so that higher order moments are negligible \citep{Dieckmann1996}. In practice, this corresponds to setting $M = N =1$ in Section~\ref{Sec:TruncationSection2}, and neglecting terms involving $m_k, m_1^k, f_k$ or $V_k$ for $ k \geq 2$. These assumptions are possibly quite strong when considering cancer evolution \citep{Aguade-Gorgorio2018}, particularly due to the role of phenotypic diversity in the population-level response to therapy. Here, we utilize the model proposed by \citet{Pressley2021} to illustrate the possible consequences of these modelling assumptions.  

\subsubsection{Application to adaptive dynamics in cancer} 
In recent work, \citet{Pressley2021} developed a mathematical framework to capture continuous adaptation to treatment, in which the cellular phenotype is considered as a direct measure of cell resistance to anti-cancer therapy. 
The authors consider a monomorphic population with population size $P(t)$ and phenotype $m_1(t)$ -- where we use this notation as the phenotype of a monomorphic population corresponds to the mean phenotypic state (i.e. the first moment) of a very sharp (Dirac delta) phenotypic distribution. Their model is given by 
\begin{equation} \label{Eq:PressleyGenericSystem}
    \left. 
\begin{aligned}
    \TimeDeriv P(t) & = P(t) G(m_1(t),P(t)) ,\\
    \TimeDeriv m_1(t) & = \alpha \left. \partial_x  G(x,y) \right|_{(x,y) = (m_1(t),P(t))} . 
\end{aligned}
\right \}
\end{equation}
In Eq.~\eqref{Eq:PressleyGenericSystem}, $\alpha$ is the speed of phenotypic adaptation and the function $G$ captures the net proliferation rate of the population $P$ in the presence of anti-cancer treatment under the assumption that all tumour cells have phenotype $m_1$. Specifically, \citet{Pressley2021} set
\begin{align*}
    G(m_1(t),P(t)) = r( m_1(t) )\left(1-\frac{P(t)}{\kappa} \right) -d -\frac{c(t)}{k+bm_1(t)} 
\end{align*}
where $\kappa$ is the carrying capacity of the population, $d$ is an intrinsic death rate in the absence of treatment, and   
\begin{align*}
    c(t) = \left \{ 
    \begin{array}{cc}
       1 & \textrm{during treatment}, \\
        0 & \textrm{otherwise},
    \end{array} 
    \right. 
\end{align*}
models the anti-cancer treatment. The treatment effect is modulated by the treatment resistance of the population with mean phenotype $m_1(t)$. The magnitude of the resistance benefit is modelled by $b$ and treatment resistance has a half-effect value of $k$. \citet{Pressley2021} used this model to quantify the benefits of adaptive therapy. 
They applied treatment until the tumour reached half the initial size, $P(t) = P(0)/2$. 
Treatment was then interrupted and withheld until the tumour reached the initial size, $P(t) = P(0)$, at which point treatment was re-applied. The model includes the cost of resistance by decreasing the intrinsic growth rate $r(m_1(t))$ of a population with mean phenotype $m_1(t)$ by
\begin{align*}
    r(m_1(t)) = r_{max}\exp(-g m_1(t)),
\end{align*}
where $r_{max}$ is the maximal growth rate and $g$ is the cost of resistance.

We now show how our framework can extend Eq.~\eqref{Eq:PressleyGenericSystem} to include population heterogeneity by considering higher order moments and setting $N=M=2$ in Eq.~\eqref{Eq:TruncatedGenericMomentODE:open}. 
\citet{Pressley2021} assume that the mean phenotype changes with direction determined by the gradient of $G$ and scaled by an adaptation speed $\alpha.$ Under these assumptions, the velocity $V(t,m_1)$ in the advection term of Eq.~\eqref{Eq:StructuredPDE} is given by  
\begin{align}\label{Eq:PressleyVelocityApproximation}
    V(t,m_1)  = \alpha \partial_x G(m_1,P).
\end{align} 
This assumption on the velocity of the advection term represents gradient ascent toward the phenotype of optimal fitness. For notational simplicity, we write for $k\in\mathbb{N}$
\begin{align*}
    \partial_x G(m_1,P) =  \left. \partial_x  G(x,y) \right|_{(x,y) = (m_1(t),P(t))}   \quad \textrm{and} \quad   \partial^k_x G(m_1,P) =  \left. \partial^k_x  G(x,y) \right|_{(x,y) = (m_1(t),P(t))} .
\end{align*}
We emphasize that this definition of $V(t,m_1)$ follows from the modelling assumptions made by \citet{Pressley2021}. However, our approximation framework applies in general to other choices for $V$ and $f$.  
Taking $N = M = 2$ -- motivated by the results of Section~\ref{sec:example:new} -- in system~\eqref{Eq:TruncatedGenericMomentODE:open} and using Gaussian closure, we obtain
\begin{equation}\label{Eq:PressleyGaussianMoments} 
\left. 
\begin{aligned} 
  \TimeDeriv  P & =  P \left[   G(m_1,P) +   \frac{\partial^2_xG(m_1,P)}{2} \left(m_2-m_1^2\right)  \right] ,  \\
 \TimeDeriv m_1 & = \left(\alpha +  \left(m_2-m_1^2\right)  \right) \partial_x G(m_1,P) \\
 & {} \quad + \left[ \frac{\alpha}{2} \partial^3_x G(m_1,P)- m_1 \frac{\partial^2_x G(m_1,P)}{2} \right] \left(m_2-m_1^2\right)  + \frac{\partial^3_x G(m_1,P)}{6} m_3 ,  \\
 \TimeDeriv m_2 & =  2 \left[ \alpha \partial_x G(m_1,P) m_1 + \alpha \partial^2_x G(m_1,P) \left(m_2-m_1^2\right) + \frac{\alpha}{2} \partial^3_x G(m_1,P)  \left(m_3-2m_1m_2+m_1^3 \right)  \right]\\
 & \quad {} 2\beta - m_2\left[ \frac{\partial^2_x G(m_1,P)}{2} \left(m_2-m_1^2\right) \right] \\
 &  \quad {} +  \partial_x G(m_1,P) \left(m_3 -m_1m_2\right) + \frac{\partial^2_x G(m_1,P)}{2} \left( m_4 - 2m_1m_3 + (m_1^2)m_2 \right) .
\end{aligned}
\right \}
\end{equation} 
where we do not explicitly show the time-dependence of the moments to simplify the notation. Similarly, we take $N = M = 2$ and use truncation closure to find
\begin{equation}\label{Eq:PressleyTruncationMoments} 
\left. 
\begin{aligned} 
  \TimeDeriv  P & =  P \left[   G(m_1,P) +   \frac{\partial^2_xG(m_1,P)}{2} \left(m_2-m_1^2\right)  \right]  ,  \\
 \TimeDeriv m_1 & = \left(\alpha + m_2-m_1^2\right)  \partial_x G(m_1,P) + \left[\alpha \frac{\partial^3_xG(m_1,P)}{2}  - 3m_1(t) \frac{\partial^2_xG(m_1,P)}{2} \right](m_2-m_1^2)  ,  \\
\TimeDeriv m_2 & = 2\beta + 2\alpha  \partial_xG(m_1,P) m_1 -G(m_1,P) m_2  + \left[  2 \alpha \partial^2_xG(m_1,P) - m_2 \frac{\partial^2_xG(m_1,P)}{2}  \right] (m_2-m_1^2)  . 
\end{aligned}
\right \}
\end{equation}

In both Eq.~\eqref{Eq:PressleyGaussianMoments} and Eq.~\eqref{Eq:PressleyTruncationMoments}, the population growth rate now depends not only on the fitness of the mean phenotype, $G(m_1,P),$ but also on it's curvature and the variance of the population about the mean phenotype, $\sigma^2 = m_2-m_1^2.$ 
Further, the differential equation for $m_1(t)$ links the phenotypic evolution with $\sigma^2$ in a similar manner to \citet{Dieckmann1996} in the context of genetic mutation.

To illustrate the impact of including phenotypic heterogeneity, we simulated the adaptive therapy regiment in Eq.~\eqref{Eq:PressleyGenericSystem}, Eq.~\eqref{Eq:PressleyGaussianMoments} and Eq.~\eqref{Eq:PressleyTruncationMoments}. We set $\alpha = 0.0001$ and otherwise used the same parameter values and initial conditions as \citet{Pressley2021}, with $P(0) = 6000$ and $m_1(0) = 0$. In Eq.~\eqref{Eq:PressleyGaussianMoments} and Eq.~\eqref{Eq:PressleyTruncationMoments}, we set $m_2(0) = m_1(0)^2$ to represent an initially monomorphic population, and set $\beta = 0.0002$. 
In Panels A and B of Fig~\ref{Fig:PressleyModelExample}, we show the predicted tumour population during adaptive therapy.  Eq.~\eqref{Eq:PressleyGenericSystem} predicts long-term tumour control, with the Time-To-Progression -- i.e. the first time $t=t_{TTP}$ at which $P(t_{TTP}) = 0.7\kappa$ -- of at least 800 days. 
Conversely, the models that include population heterogeneity, i.e. Eq.~\eqref{Eq:PressleyGaussianMoments} and Eq.~\eqref{Eq:PressleyTruncationMoments}, both predict the evolution of treatment resistance and failure of adaptive therapy within 28 days. In Panel C of Fig~\ref{Fig:PressleyModelExample}, we show the cumulative dose during adaptive treatment, which highlights the development of resistance predicted by both Eq.~\eqref{Eq:PressleyGaussianMoments} and Eq.~\eqref{Eq:PressleyTruncationMoments}.  
This highlights that population diversity may accelerate the development of treatment resistance as has been observed in many experimental \citep{Bodi2017,McGranahan2017,Sottoriva2013,Marine2020,Craig2019} and theoretical \citep{almeida2019evolution,ardavseva2020evolutionary,lorenzi2015dissecting,villa2021evolutionary,Kohn-Luque2023,Cassidy2018,Lavi2013,Greene2014,Nichol2016} studies. We emphasize that the parameters used in this simulation were not obtained by fitting the model to data, but rather chosen to qualitatively illustrate the influence of including population heterogeneity in a model of adaptive evolution of treatment resistance.

\begin{figure}[!ht]
\centering 
\includegraphics[scale=0.93, trim= 5 5 5 5,clip]{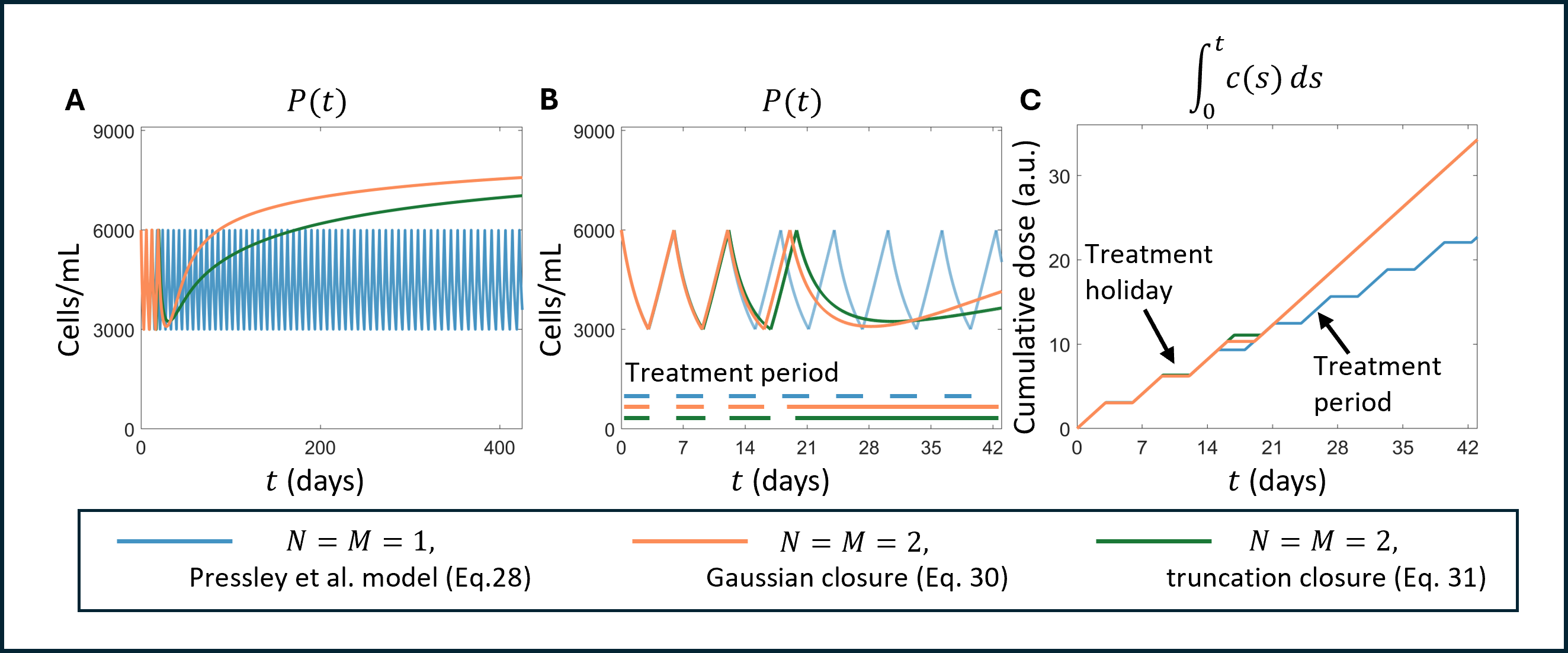}
\caption{
Population heterogeneity drives treatment resistance in the adaptive dynamics model implemented by \citet{Pressley2021}. Panels A and B show the dynamics of the simulated tumour population with and without heterogeneity. The horizontal lines in Panel B show when treatment is applied. Panel C shows the cumulative dose administered during the adaptive therapy treatment for each simulation. In all cases, the $N=M=1$ approximation, corresponding to the Pressley model in Eq.~\eqref{Eq:PressleyGenericSystem}, is shown in blue while the orange and green curves show the $N = M = 2$ approximation models with Gaussian and truncation closures, i.e. the models given by Eq.~\eqref{Eq:PressleyGaussianMoments} and Eq.~\eqref{Eq:PressleyTruncationMoments} respectively. Other than the adaptation speed, $\alpha = 0.00001$, and diffusion coefficient $\beta = 0.00002$, all simulations used the parameter values given in \citet{Pressley2021}.  }  

\label{Fig:PressleyModelExample}
\end{figure}

 \section{Discussion}\label{Sec:Discussion}
 
Mathematical models have become increasingly important in our understanding of the mechanisms driving cancer evolutionary dynamics and the role of intratumour phenotypic heterogeneity. 
Phenotype-structured PDE models describe the temporal dynamics of both the tumour size and phenotypic composition.
In the context of adaptive dynamics in continuously structured phenotype space, these PDE models carry a distinct advantage in that individual model terms are directly biologically interpretable. 

\paragraph{Benefits of reducing a phenotype-structured PDE to a system of ODEs.}
While phenotype-structured PDE models are well-poised to interpret biological data, they carry, by their very nature, challenges that do not apply to the ODE models that are routine in the interpretation of data. For example, for a mathematical model to be uniquely parameterised it must be \textit{identifiable}~\citep{Bellman.1970,Raue:2009,Cassidy2023a}. Outside of recent work \cite{browning2024structural,Boiger2016,Stapor2018a}, most computational tools for parameter identifiability analysis are developed for ODE models~\citep{Bellman.1970,Raue:2009}. Therefore, assessing parameter identifiability in PDE models often requires an ODE surrogate such as an equivalent or approximate system of moment equations \cite{Browning.2020}. 
Parameter identifiability is particularly pertinent in the case of phenotype-structured PDE models, as determining the phenotypic distribution of a tumour sample is experimentally demanding \citep{brestoff2022contemporary,Pigliucci2010}.
Furthermore, numerical methods for phenotype-structured PDEs are not widely implemented in existing scientific software \citep{Carpenter2017}. Consequently, fitting these PDE models to experimental data often requires the development of problem specific software. The resulting numerical methods are typically more computationally expensive than for ODE models. As parameter estimation typically involves many simulations of the mathematical model which amplifies the increased computational cost of solving the PDE. Therefore, the ability to reduce a phenotype-structured PDE to a system of ODEs while maintaining the ability to characterise the tumour composition allows for significant computational efficiency.

\paragraph{Summary and novelty of our model-reduction procedure.}
Accordingly, we proposed a generalised method to reduce a phenotype-structured PDE model of cancer adaptive dynamics to a system of ODEs for the moments of the phenotypic distribution, up to an arbitrary number of moments. This reduction allows modellers to use existing technical tools for ODE models while maintaining the biological relevance and the interpretability of the phenotype-structured PDE. 
The model-reduction procedure we propose relies on the use of the moment generating function of the phenotypic distribution, a Taylor series expansion of the phenotypic drift and proliferation rate functions, and truncation closure. Our work extends the analysis of \citet{almeida2019evolution} and \citet{chisholm2016evolutionary}, to a more biologically relevant phenotypic domain and a model without any \textit{a priori} assumptions on the shape of the distribution or the dependency of the phenotypic drift and proliferation rate on the phenotypic trait.
Our work also extends the analysis of~\citet{Dieckmann1996} that ties a stochastic model of mutation-selection dynamics to the adaptive dynamics models in Section~\ref{Sec:PressleyExample}. Here, we have used both \textit{truncation} and \textit{Gaussian} closures for the resulting system of ODEs describing the moments of the phenotypic distribution $\phat$. However, other approaches are possible and evaluating these alternative closures is an obvious area for future research.

\paragraph{Strengths and limitations.} 
Our model reduction procedure is independent of both the shape of the phenotypic distribution and the functional form of the terms that characterise adaptation.
This generality and flexibility lends our analysis suitable for adaptation in a wide range of contexts, both within mathematical oncology and more broadly. We expect many of the advantages conferred by the reduced model to become even more pertinent in high-dimensional phenotype spaces, particularly in the context of mechanistic interpretation of correspondingly high-dimensional multi-omics data. The necessity to impose a system closure yields an approximation of the underlying dynamics. However, we highlight that the presented approach can be applied up to an arbitrary order. The question of how many moments a problem requires, the closure to apply, or the effect of closure on parameter identifiability, remains open even in fields where moment closures have a long history \cite{smith2007quadrature,Kuehn:2016uf,ghusinga2017approximate,Schnoerr.2017wbb,Browning.2020,wagner2022quasi}. We expect unimodal phenotypic distributions to be well characterised by lower-order moments. For high-dimensional problems, the question of closure type and order is likely determined by computational cost. Ultimately, we provide a general and flexible framework for describing adaptation in a continuously-distributed phenotype space whilst retaining the computational and analytical advantages of ODE-based approaches. 

\paragraph{Consequences and perspectives in cancer adaptive therapy.}
Our work relaxes the assumptions of near-linear growth rates or vanishing variance underlying the canonical equation of adaptive dynamics \citep{Dieckmann1996} by explicitly linking the higher-order moments of the population distribution in phenotype space with the resulting impacts on population growth and adaptation. 
To illustrate the possible effects of including population heterogeneity, we considered the model of adaptive dynamics in response to adaptive therapy by \citet{Pressley2021}. 
We show that including phenotypic heterogeneity in the \citet{Pressley2021} model can lead to failure of an adaptive therapy strategy that would otherwise result in long-term tumour control. While this result is unsurprising \citep{Bodi2017}, it illustrates how phenotypic heterogeneity can drive treatment resistance \citep{Aguade-Gorgorio2018,Hanahan2022}. Consequently, our results illustrate a simple way to introduce population heterogeneity in ODE models investigating the development of treatment resistance, and it's importance in the emergent model dynamics

In this context, our work is directly relevant to recent multi-omics-level experiments characterising cellular phenotypes. However, these high dimensional data sets are challenging to interpret and thus numerous dimensionality reduction methods have been proposed~\cite{Oshternian2024,Burkhardt2022}. Amongst these are phenotype classification methods that summarise these data sets with a selected number of moments~\cite{Vogelstein2021,tang2010}. As such, moment-based descriptions are increasingly considered to describe phenotypic states and our modelling therefore offers a direct link with these emerging clinically relevant data sets. For example, bulk and single-cell sequencing have both identified continuous phenotypic adaptation in response to treatments in patient-derived xenografts carrying the BRAFV600E mutation~\cite{Xue2017}.
\citet{Kohn-Luque2023} have shown how to use \textit{in vitro} dose-response experiments to parameterize mathematical models of phenotypically distinct sub-populations through the \textit{PhenoPop} method, corresponding to a discrete phenotype distribution. The resulting models capture the sub-population dynamics but assume that each sub-population is homogeneous. Our results complement this approach by allowing for the \textit{PhenoPop} method to also describe the population heterogeneity. 

Consequently, the framework derived in this work could facilitate both the development of mathematical models and the calibration of these models by multi-omics data to understand how phenotypic heterogeneity drives the evolution of treatment resistance to targeted therapies. 

\section*{Code availability}

The Matlab code underlying the simulations in Section~\ref{sec:example:new} (which allows to test the approximation for custom $f$ and $V$ for various $N$ and $M$) is available at \texttt{https://github.com/ChiaraVilla/Villa2025Reducing}, while the code underlying the simulations in Section~\ref{Sec:PressleyExample} is available at   \texttt{https://github.com/ttcassid/\allowbreak /Phenotype$\_$Continuous$\_$Adaptation}.

\section*{Acknowledgments}
This work was partially supported by a Heilbronn Institute for Mathematical Research Small Maths Grant to TC. APB thanks the Mathematical Institute, Oxford for a Hooke Research Fellowship. ALJ thanks the London Mathematical Society. SH was funded by Wenner-Gren Stiftelserna/the Wenner-Gren Foundations (WGF2022-0044) and the Kjell och M{\"a}rta Beijer Foundation. This project was partially supported by the European Union's Horizon 2020 research and innovation programme under the Marie Sk\l{}odowska-Curie grant agreement No 945298-ParisRegionFP. CV is a Fellow of the Paris Region Fellowship Programme, supported by the Paris Region.

\section*{Compliance with Ethical Standards}
No potential conflicts of interest to declare.

\appendix

\section{Proof of Proposition~\ref{lemma:m0}}\label{app:proof:prop}

    Integrating \eqref{Eq:StructuredPDE} with respect to $x$ over $\Omega$, and interchanging integration and differentiation, gives
\begin{equation*}
\TimeDeriv P(t) = \int_\Omega \partial_x [ \beta \partial_x p(t,x) - V(t,x)p(t,x)] \d x +  \int_\Omega \left( f(t,x)- \frac{P(t)}{\kappa} \right) p(t,x) \d x ,
\end{equation*}
where we have also used definition~\eqref{Eq:PopulationSize}. Applying the boundary condition~\eqref{Eq:BCs}, and again using~\eqref{Eq:PopulationSize}, immediately gives Eq.~\eqref{Eq:PopulationSizeODE}. The
initial condition~\eqref{ic:m0} can be obtained by integrating the initial condition~\eqref{Eq:ICs} directly. The strict positivity of $P(0)$ follows directly from the assumption imposed on $p^0(x)$ in~\eqref{Eq:ICs}. 
Under assumption~\eqref{ass:fv:bounded}, from~\eqref{Eq:PopulationSizeODE} the following inequality holds
\begin{align*}
    \TimeDeriv P(t)&\leq f_\text{M}\int_\Omega p(t,x) \d x - \frac{P^2(t)}{\kappa} \leq \left( f_\text{M} - \frac{P(t)}{\kappa} \right) P(t)
\end{align*}
which, setting $\overline{P}:=  \max (P(0), f_\text{M}\kappa)$, gives the upper bound in~\eqref{m0:bound}. Similarly, we have
\begin{align*}
    \TimeDeriv P(t)&\geq f_\text{m}\int_\Omega p(t,x) \d x - \frac{P^2(t)}{\kappa} \geq \left( f_\text{m} - \frac{P(t)}{\kappa} \right) P(t). 
\end{align*}
We note that the lower bound for $\TimeDeriv P(t)$ is a scalar logisitic differential equation. Thus, Gronwall's inequality immediately yields the strict positivity of $P(t)$ for all $t>0$. \hfill \qed

\section{Full system under truncation closure}
\setcounter{equation}{0}\renewcommand\theequation{A\arabic{equation}}

System~\eqref{Eq:TruncatedGenericMomentODE:open} under truncation closure~\eqref{closure:truncation} can be rewritten as
\begin{equation}\label{Eq:TruncatedGenericMomentODE}
\left. 
\begin{aligned} 
\TimeDeriv  P(t) & =  P(t) \displaystyle \sum_{n=0}^{M} f_n(t) \left[ \displaystyle \sum_{i=\max (0,n-N)}^n  (-1)^{i}  \binom{n}{i} (m_1(t))^{i} m_{n-i}(t) \right] - \frac{P^2(t)}{\kappa} , \\
\TimeDeriv m_1(t) & =  \displaystyle \sum_{n=0}^{M} \left(V_n(t)  - m_1(t)f_n(t) \right) \left[ \displaystyle \sum_{i=\max (0,n-N)}^n (-1)^{i} \binom{n}{i} (m_1(t))^{i}  m_{n-i}(t) \right]  \\
& {} \quad + \displaystyle \sum_{n=0}^{M}  f_n(t) \left[  \displaystyle \sum_{i=\max (0,n+1-N)}^n  (-1)^{i} \binom{n}{i} (m_1(t))^{i} m_{n+1-i}(t) \right] , \\
\TimeDeriv m_k(t) & = - m_k(t) \displaystyle \sum_{n=0}^{M} f_n(t)  \displaystyle \left[ \sum_{i=\max (0,n-N)}^n  (-1)^{i}  \binom{n}{i} (m_1(t))^{i} m_{n-i}(t) \right]   + \beta  k(k-1) m_{k-2}(t)  \\
    &   \quad + k \displaystyle \sum_{n=0}^{M} V_n(t)  \left[ \displaystyle \sum_{i=\max (0,n+(k-1)-N)}^n  (-1)^{i} \binom{n}{i} (m_1(t))^{i} m_{n+(k-1)-i}(t)  \right]    \\
    & \quad  + \displaystyle \sum_{n=0}^{M} f_n(t)  \left[  \displaystyle \sum_{i=\max (0,n+k-N)}^n  (-1)^{i} \binom{n}{i} (m_1(t))^{i} m_{n+k-i}(t) \right] , \quad 2 \leq k \leq N .   
\end{aligned}
\right \}
\end{equation}
\section{Supplementary figures}
\setcounter{figure}{0}\renewcommand\thefigure{C\arabic{figure}}

\begin{figure}[htb!]
    \centering
    \includegraphics[width=1\linewidth]{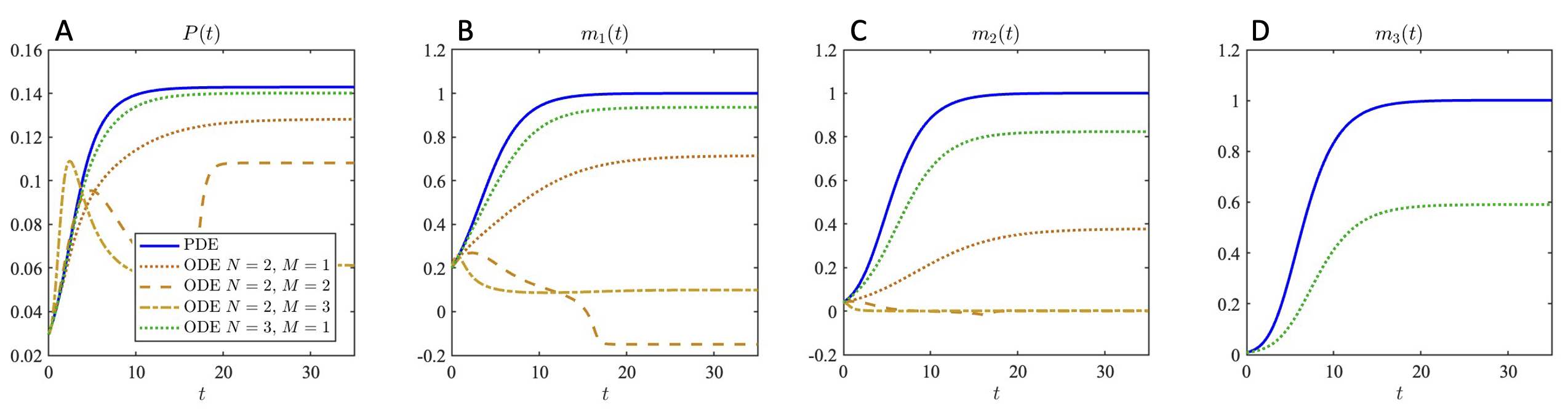}
    \caption{Comparison between the numerical solution of PDE~\eqref{Eq:StructuredPDE}, complemented with Eq.~\eqref{Eq:PopulationSize} and boundary conditions~\eqref{Eq:BCs}, and that of the ODE system~\eqref{Eq:TruncatedGenericMomentODE:open} under truncation closure~\eqref{closure:truncation}, for $f$ and $V$ defined in~\eqref{ex:fV:adhoc} ($\omega=2$) and initial condition~\eqref{ic:truncatednormal}. The moments corresponding to the phenotypic density function $p(t,x)$ obtained solving the PDE are plot with a blue solid line: the populations size $P(t)$ (panel A), the first moment $m_1(t)$ (B), the second moment $m_2(t)$ (C) and the third moment $m_3(t)$ (D). The moments obtained from the ODE system are also shown in each panel under different choices of $N$ and $M$: $N=2$ and $M=1$ (brown dotted line), $N=2$ and $M=2$ (brown dashed line), $N=2$ and $M=3$ (brown dash-dotted line), $N=3$ and $M=1$ (green dotted line), $N=3$ and $M=1$ (green dashed line). The results correspond to the parameter set: $f_{max}=2$, $k_m=0.4$, $V_{max}=0.5$, $\omega=2$, $\kappa=0.1$, $\beta=10^{-4}$, $\bar{x}_0=0.2$, $\sigma_0=0.02$, $l=0$, $L=1.2$. The numerical schemes rely on a first order forward difference approximation of the time derivatives, second order central difference for the diffusion term and first order upwind for the advection term. }  
    \label{fig:ex2:truncation}
\end{figure}
\begin{figure}[htb!]
    \centering
    \includegraphics[width=1\linewidth]{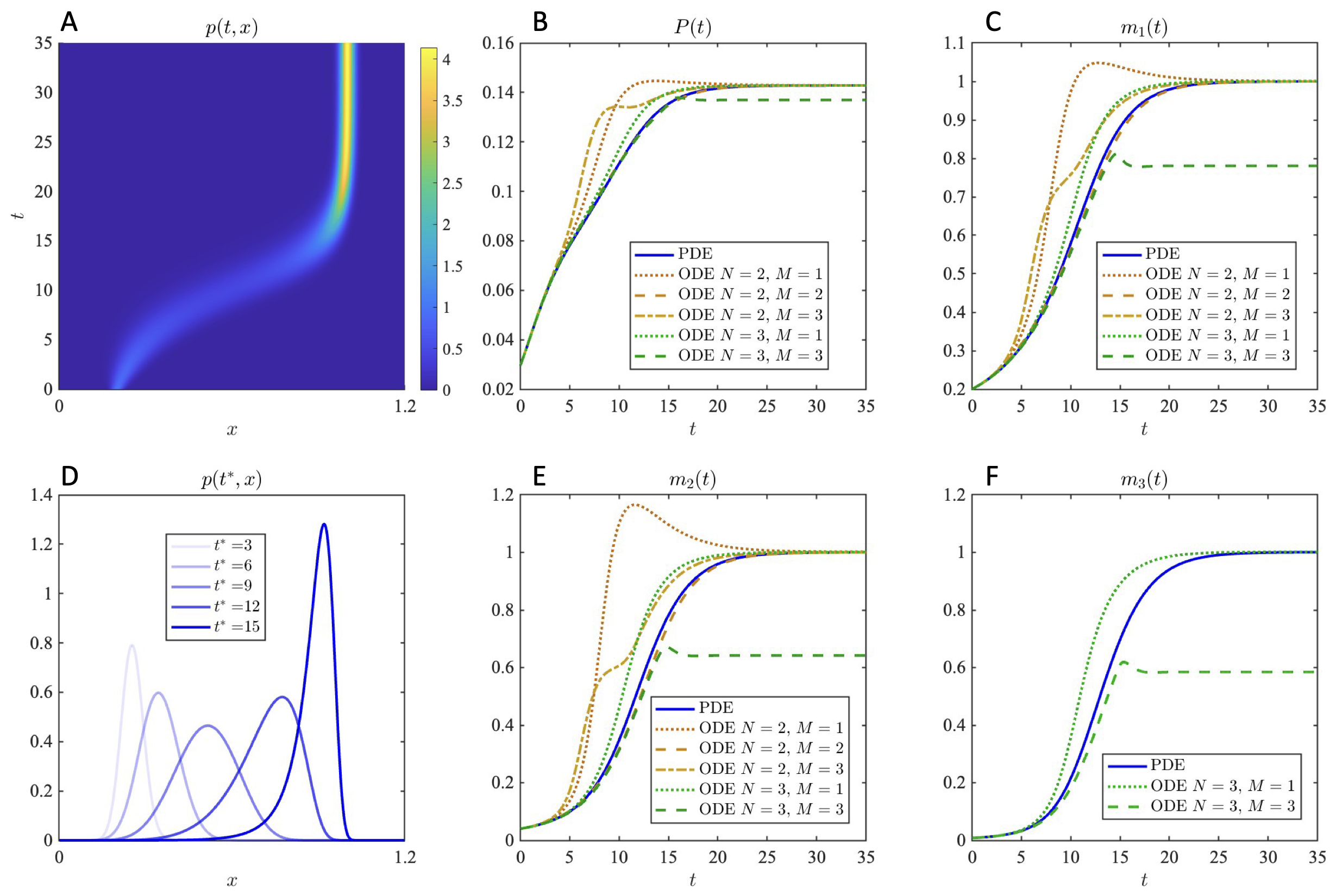}
     \caption{Comparison between the numerical solution of PDE~\eqref{Eq:StructuredPDE}, complemented with Eq.~\eqref{Eq:PopulationSize} and boundary conditions~\eqref{Eq:BCs}, and that of the ODE system~\eqref{Eq:TruncatedGenericMomentODE:open} under Gaussian closure~\eqref{closure:Gaussian}, for $f$ and $V$ defined in~\eqref{ex:fV:adhoc} ($\omega=2$) and initial condition~\eqref{ic:truncatednormal}. The phenotypic density function $p(t,x)$ obtained solving the PDE is shown in the first column for $t\in[0,35]$ (panel A) and at selected times $t^*=3,6,9,12,15$ (D). The corresponding moments are plot with a blue solid line in the remaining panels: the populations size $P(t)$ (B), the first moment $m_1(t)$ (C), the second moment $m_2(t)$ (E) and the third moment $m_3(t)$ (F). The moments obtained from the ODE system are also shown in each panel under different choices of $N$ and $M$: $N=2$ and $M=1$ (brown dotted line), $N=2$ and $M=2$ (brown dashed line), $N=2$ and $M=3$ (brown dash-dotted line), $N=3$ and $M=1$ (green dotted line), $N=3$ and $M=1$ (green dashed line). The results correspond to the parameter set: $f_{max}=2$, $k_m=0.4$, $V_{max}=0.5$, $\omega=2$, $\kappa=0.1$, $\beta=10^{-4}$, $\bar{x}_0=0.2$, $\sigma_0=0.02$, $l=0$, $L=1.2$. The numerical schemes rely on a first order forward difference approximation of the time derivatives, second order central difference for the diffusion term and first order upwind for the advection term.   }
    \label{fig:ex2:om2}
\end{figure}

\end{document}